\documentclass[aps,english,notitlepage,nofootinbib]{revtex4}

\usepackage{amstext,amsmath,amssymb,amsfonts,bbm}
\usepackage{epsfig}
\usepackage{hyperref}
\usepackage{amsthm}
\usepackage{subfig}
\usepackage{wrapfig}
\usepackage{color}
\usepackage{multirow}
\usepackage{soul}

\usepackage[latin9]{inputenc}
\usepackage{babel}

\usepackage{enumitem}

\usepackage[font={footnotesize,it}]{caption}
\usepackage{titlesec}
\titleformat*{\paragraph}{\bf\small}

\usepackage{array}
\usepackage{mathtools}
\usepackage{amsmath}
\usepackage{amsthm}
\usepackage{amssymb}
\usepackage{cancel}

\theoremstyle{plain}
\newtheorem{lem}{\protect\lemmaname}
\theoremstyle{plain}

\theoremstyle{definition}
\newtheorem{defn}{\protect\definitionname}

\theoremstyle{plain}
\newtheorem{thm}{\protect\theoremname}
\theoremstyle{plain}
\newtheorem*{lem*}{\protect\lemmaname}
\theoremstyle{plain}
\newtheorem*{cor*}{\protect\corollaryname}
\theoremstyle{plain}
\newtheorem*{thm*}{\protect\theoremname}
\theoremstyle{plain}
\newtheorem{prop}{Proposition}

\providecommand{\corollaryname}{Corollary}
\providecommand{\definitionname}{Definition}
\providecommand{\lemmaname}{Lemma}

\providecommand{\theoremname}{Theorem}

\usepackage{import}
\usepackage{tikz}
\usetikzlibrary{tikzmark}

\newcommand{\bea}{\begin{eqnarray}}
\newcommand{\eea}{\end{eqnarray}}

\newcommand{\be}{\begin{equation}}
\newcommand{\ee}{\end{equation}}

\begin{document}
\begin{flushleft}
KCL-PH-TH/2018-7
\par\end{flushleft}

\date{\today}

\title{Minimizers of the dynamical Boulatov model}

\author{Joseph Ben Geloun}
\email{jobengeloun@lipn.univ-paris13.fr}

\affiliation{LIPN, UMR CNRS 7030, Institut Galil\'ee, Universit\'e Paris 13, Sorbonne Paris Cit\'e, 99, avenue Jean-Baptiste Cl\'ement, 93430 Villetaneuse, France}

\affiliation{International Chair in Mathematical Physics and Applications, ICMPA-UNESCO Chair, 072Bp50, Cotonou, Benin}

\author{Alexander Kegeles}
\email{kegeles@aei.mpg.de}

\affiliation{Max-Planck-Institute for Gravitational Physics (Albert-Einstein-Institute),\\
Am  M\"uhlenberg 1, 14476 Potsdam-Golm, Germany}

\author{Andreas G. A. Pithis}
\email{andreas.pithis@kcl.ac.uk}

\affiliation{Department of Physics, King's College London, University of London, \\
Strand WC2R 2LS, London, United Kingdom}

\affiliation{Max-Planck-Institute for Gravitational Physics (Albert-Einstein-Institute),\\
Am  M\"uhlenberg 1, 14476 Potsdam-Golm, Germany}

\begin{abstract}
We study the Euler-Lagrange equation of the dynamical Boulatov model which is a simplicial model for 3d Euclidean quantum gravity augmented by a Laplace-Beltrami operator. We provide all its solutions on the space of left and right invariant functions that render the interaction of the model an equilateral tetrahedron. Surprisingly, for a non-linear equation of motion, the solution space forms a vector space. This space distinguishes three classes of solutions: saddle points, global and local minima of the action. Our analysis shows that there exists one parameter region of coupling constants for which the action admits degenerate global minima.
\end{abstract}


\maketitle

\section{Introduction}

In three dimensions, general relativity can be formulated as a $BF$-theory~\citep{Baez:1999sr}. Its functional integral quantization discretized over simplicial complexes leads to the Ponzano-Regge model~\cite{Ponzano:1969aa,Barrett:2008wh}, which can be regarded as a quantum gravity model of discrete geometry. A cornerstone of this approach then is to recover continuum geometries with all desired requirements and properties of a three-dimensional spacetime. Such a description, however, as well as a mechanism, which could successfully sort it, remains an open problem in background independent approaches to quantum gravity.

The Boulatov model of group field theory (GFT)~\citep{Freidel:2005qe,Oriti:2006se} provides one way to address this issue for Euclidean signature. The model is formally defined by the generating functional,
\begin{equation}\label{Z}
  Z\left[J\right]
  =\int\mathcal{D}\varphi\,e^{-S\left(\varphi\right)+\int J\varphi},
\end{equation}
where $S\left(\varphi\right)$ denotes the Boulatov action~\citep{Boulatov:1992vp}.
The striking fact about this generating functional is that its Feynman graphs correspond to simplicial complexes and its Feynman amplitudes coincide with Ponzano-Regge spin foam amplitudes~\cite{Ponzano:1969aa,Barrett:2008wh}.
One concludes that a perturbative expansion of Eq.~\eqref{Z} provides a discrete model of quantum gravity and that a description of continuum geometries will require a non-perturbative understanding of Eq.~\eqref{Z}.

The construction of a full non-perturbative quantum field theory is rarely possible, but often it is already enough to construct a perturbation theory around a non-perturbative vacuum~\citep{Dittrich:1985tr}.
Moreover, if quantum fluctuations are not too strong, a non-perturbative vacuum can be reasonably well approximated by the minimum of the classical action $S$, called the minimizer. In that case, the mean-field approximation around the minimizer prompts an effective field theory that will capture the non-perturbative regime of the model. For that reason, a study of minimizers of the Boulatov action is an important step towards a better understanding of continuous quantum geometries.

Despite their importance, however, the extrema of the Boulatov action are poorly understood in the literature. This is mostly because the Euler-Lagrange equations of the Boulatov action are non-linear differential equations that also involve integrals. Such equations are called integro-differential equations; generally, they are notoriously difficult to solve. In the Boulatov model, these integro-differential equations can be formulated as integral equations with an integral kernel given by the Wigner $6J$-symbol. A solution of the extremal equations then requires full control of the zeros of the $6J$-symbol, which remains an open problem despite decades of research~\cite{Lindner:1985aa,brudno1985nontrivial,brudno1987nontrivial,A-Heim:2009aa}. This makes the complete analysis of the  problem out of reach.

In addition to this, there seems to be no consensus on the signs of the coupling constants in GFT models. For instance, the convention used in renormalization analyses~\citep{Carrozza:2016vsq} is opposite to the one used in the context of the GFT condensate cosmology investigations~\citep{Gielen:2013kla,Gielen:2013naa,Gielen:2014ila,Oriti:2016qtz,Oriti:2016ueo,deCesare:2016rsf,Pithis:2016wzf,Pithis:2016cxg,deCesare:2017ynn}. Despite this ambiguity in the sign convention both analyses rely on the existence of global or at least local minimizers and for that reason require a good understanding of the extrema in GFT.

In this work, we address the minimizers of the Boulatov action augmented by a Laplace-Beltrami operator, hereafter called dynamical Boulatov action~\citep{Geloun:2013zka}. To make the problem tractable, we look for minimizers in the space of left and right invariant fields corresponding to equilateral triangles.  Section~\ref{sec:Boulatovmodel} gives the definition of the model and the space of functions considered in this article. On this space the Euler-Lagrange equations of the action become solvable, allowing us to provide a full characterization of solutions in section~\ref{subsec:Extrema-of-the action}.
We then identify the parameter regimes in which the action admits minima and characterize the minimizers in section~\ref{subsec:Minimizer-of-the action}.
Our main result regarding the extrema is presented in theorem~\ref{thm:extrema of dynamical Boulatov action} and the subsequent discussion. The characterization of minimizers is provided in theorem~\ref{thm:minimizeres}. Implications of our results on the quantum theory are discussed in section~\ref{sec:discussion}. A closing appendix gathers useful identities and the proofs of some statements in the text.

\section{The dynamical Boulatov model}
\label{sec:Boulatovmodel}

This section reviews the construction of the Boulatov model and
sets up our notations. We assume the reader to be familiar with the harmonic
analysis on $\mathrm{SU}\left(2\right)$ (appendix~\ref{sec:appendix}
however, reports useful notions on this topic).

Let $\mathcal{C}^{\infty}\left(M\right)$ be the space of smooth, real-valued functions defined on the compact Lie group $M=\mathrm{SU}\left(2\right)^{\times3}$.
The components of elements of $M$ are denoted by a subindex such that
$ x=\left(x_{1},x_{2},x_{3}\right)\in M $.
Define the space  $\mathcal{S}$ of right and cyclic invariant functions. That
is functions $f$ in   $\mathcal{C}^{\infty} \left( M \right)$ that satisfy  {\bf right invariance}: for any $R\in\mathrm{SU}\left(2\right)$ and
any $x\in M$, $f\left(x_{1}R,x_{2}R,x_{3}R\right)=f\left(x_{1},x_{2},x_{3}\right);$ and {\bf cyclicity}: for any $x\in M$,
$  f\left(x_{1},x_{2},x_{3}\right)=f\left(x_{2},x_{3},x_{1}\right)=f\left(x_{3},x_{1},x_{2}\right)$.

\noindent The dynamical Boulatov action is a functional $S_{m,\lambda}$ on $\mathcal{S}$, given by the integral
\bea\label{eq:action}
  S_{m,\lambda} \left( \varphi \right)=
   \frac{1}{2} \int_{M} \text{d}x~
    \varphi \left( x \right)
    \left( -\Delta+m^{2} \right)
    \varphi \left( x \right)
   + \frac{\lambda}{4!}
    \int_{M^{\times4}} \text{d}x \text{d}y \text{d}z \text{d}w~
    \text{Tet} \left( x,y,z,w \right)  \,
   \varphi\left(x\right)\varphi\left(y\right)\varphi\left(z\right)
    \varphi\left(w\right),
    \eea
where $m^2$ and $\lambda$ are real, possibly negative, coupling constants, $\text{d}x$ is the Haar measure on $M$, $\Delta$ is the Laplace-Beltrami operator on $M$ with the canonical metric,\footnote{We include the Laplace-Beltrami operator in the action, for a consistent implementation of a renormalization scheme~\citep{Geloun:2013zka}.}  and  the integral kernel \text{Tet} is given by
\be
  \text{Tet}\left(x,y,z,w\right)
 =
    \delta\left(x_{1}y_{1}^{-1}\right)\delta\left(x_{2}z_{1}^{-1}\right)\delta\left(x_{3}w_{1}^{-1}\right)
    \delta\left(y_{2}w_{3}^{-1}\right)\delta\left(y_{3}z_{2}^{-1}\right)
\delta\left(z_{3}w_{2}^{-1}\right).
\ee
This kernel encodes the combinatorics of a tetrahedron (Fig.~\ref{pic:tet}) and is symmetric under cyclic permutations of its arguments.
\begin{figure}[h!]
 
  \def\svgwidth{0.3\columnwidth}
  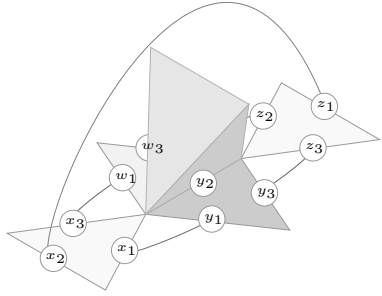
\caption{Combinatorics of a tetrahedron}
  \label{pic:tet}
\end{figure}

To address the variational problem, we  topologize $\mathcal{S}$ by the family of semi-norms
$  \|f\|_{n}\doteq\sup_{x\in M}\left|\Delta^{n}f\left(x\right)\right|$, with the neighborhood base given by semi-balls~\citep{sugiura1971},
$ N_{\epsilon,n}\left(0\right)=\left\{ \|f\|_{n}<\epsilon\,\vert\,f\in\mathcal{S}\right\} ,$
  for $n\in\mathbb{N}$ and $\epsilon>0$.

Leading the  analysis further, we will restrict the space $\mathcal{S}$
by requiring $(1)$ {\bf left invariance:}  for any $L\in\mathrm{SU}\left(2\right)$, $x\in M$ and $f\in\mathcal{S}$, $f\left(Lx_{1},Lx_{2},Lx_{3}\right)=f\left(x_{1},x_{2},x_{3}\right)$. By the Peter-Weyl theorem, every $f\in\mathcal{S}$ can then be written as
\begin{equation}
  f\left(x\right)=\sum_{J\in\mathfrak{J}}f^{J}\,\mathcal{X}^{J}\left(x\right),
  \qquad  \mathcal{X}^{J}\left(x\right) \doteq \frac{\sqrt{d_{j_{1}}d_{j_{2}}d_{j_{3}}}}{3}
\int\text{d}h\: \sum_{\sigma\in\text{Cycl}}\chi^{j_{\sigma(1)}}\left(x_{1}h\right)\chi^{j_{\sigma(2)}}\left(x_{2}h\right)\chi^{j_{\sigma(3)}}\left(x_{3}h\right),
  \label{eq:Peter-Weyl decomposition on S}
\end{equation}
where $J=\left(j_{1},j_{2},j_{3}\right)$ belongs to $\mathfrak{J}\doteq\left(\frac{\mathbb{N}}{2}\right)^{\times3}$, Cycl denotes cyclic permutations of the set \{1,2,3\} and
$\chi^{j_{i}}$ denotes the character of an $\mathrm{SU} \left( 2 \right)$ representation of
dimension $d_{j}=2j+1$ for $j\in \frac{\mathbb{N}}{2}$. By~\citep[theorem 3]{sugiura1971} the sequence of coefficients $\left(f^{J}\right)_{ J\in \mathfrak{J}}$ is a rapidly decreasing sequence of real numbers and the equality is understood such that the right hand side of the equation converges to $f\left(x\right)$ in the aforementioned topology.

Furthermore, for such functions we require $(2)$ the {\bf equilateral condition:}  let $f\in\mathcal{S}$, then $f$ is an
equilateral function if its non-vanishing Peter-Weyl coefficients
are of the form $\left(f^{\left(j,j,j\right)}\right)_{j\in\mathbb{N}}$.

We denote the restriction of $\mathcal{S}$ to left invariant equilateral
functions by $\mathcal{S}_{\text{EL}}$ and the space of equilateral
triples by $\mathfrak{J}_{\text{EL}}=\left\{ \left(j,j,j\right)\,\vert\,j\in\mathbb{N}\right\} $.
Note, that $\mathfrak{J}_{\text{EL}}$ contains only
integer multi-indices, since for any half-integer $j$ the matrix
coefficients vanish:
\begin{align*}
\mathcal{X}^{\left(j,j,j\right)} & =0\quad\text{with }j=\frac{2n+1}{2}\quad n\in\mathbb{N}.
\end{align*}
In the following we will sometimes use the notation $f\in\mathcal{S}_{\left(\text{EL}\right)}$
and $f^{J}$ with $J\in\mathfrak{J}_{\left(EL\right)}$ to signal
that the statement holds equally for $\mathcal{S}$ and
$\mathcal{S}_{\text{EL}}$ and, correspondingly, with a set of indices
belonging to $\mathfrak{J}$ or to $\mathfrak{J}_{\text{EL}}$. For clearer notation we also define the square of the triple $J$ as
$
J^{2}  \doteq j_{1}\left(j_{1}+1\right)+j_{2}\left(j_{2}+1\right)+j_{3}\left(j_{3}+1\right),
$ and its modulus as
$ \left|J\right|\doteq j_{1}+j_{2}+j_{3}. $

\begin{defn}
\label{def:minimizer}
A local minimizer of the action $S_{m,\lambda}$ on $\mathcal{S}_{\text{EL}}$ is a field $\varphi\in\mathcal{S}_{\text{EL}}$, that for some $n\in\mathbb{N}$ and $\epsilon>0$ satisfies
\begin{equation}
S_{m,\lambda}\left(\phi\right)\geq S_{m,\lambda}\left(\varphi\right),
\label{eq:minimizer condition}
\end{equation}
for any $\phi\in N_{\epsilon,n}\left(\varphi\right)\cap\mathcal{S}_{\text{EL}}$.
If condition~\eqref{eq:minimizer condition} is satisfied on the whole space $\mathcal{S}_{\text{EL}}$ we call the minimizer global.
\end{defn}
In the following we will characterize all minimizers of the action $S_{m,\lambda}$ on $\mathcal{S}_{\text{EL}}$
for the four different parameter regions
\begin{align*}
(a)\:m^{2} & <0 &
(b)\:m^{2} & >0 &
(c)\;m^{2} & >0 &
(d)\;m^{2} & <0\\
\lambda & <0 &
\lambda & <0 &
\lambda & >0 &
\lambda & >0.
\end{align*}

For each of the parameter regions, we will characterize all extrema of the action $S_{m,\lambda}$ on $\mathcal{S}_{\text{EL}}$ and identity, which of the extrema are minimizers.

We now briefly motivate the restrictions made in our analysis and point out the geometrical considerations behind the use of the space $\mathcal{S}_{\text{EL}}$.

\medskip

\noindent{\bf The space $\mathcal{S}$:}
By the Peter-Weyl theorem we can decompose any smooth field $f$ on $M$ in modes such that
$
f\left(x\right)=\sum_{ J\in \mathfrak{J}}\sum_{\alpha,\beta=-J}^{J}\,f_{\alpha,\beta}^{J}\,\mathfrak{D}_{\alpha,\beta}^{J}\left(x\right),
$
where $\mathfrak{D}^{J}=\mathfrak{D}^{(j_{1},j_{2},j_{3})}=D^{j_{1}}\otimes D^{j_{2}}\otimes D^{j_{3}},$ are the Wigner-matrix coefficients
for the product representation of $M$.
To gain intuition on the construction, we depict the Peter-Weyl
coefficients by stranded lines, emanating from a single point (Fig.~\ref{a}). Then
the right invariance of $f$ ensures a closure of the dual edges to
form a triangle (Fig.~\ref{b}). Hence, the right invariance
is necessary to give a geometric interpretation to the fields and
it is thus crucial for the connection between the Boulatov group field
theory and the Ponzano-Regge spin-foam model~\citep{Boulatov:1992vp,Baez:1999sr,Ponzano:1969aa}. In addition, the invariance under cyclic relabeling of the field arguments ensures that the ordering of the field arguments has no physical meaning.\footnote{Notice that by imposing invariance with respect to cyclic permutations of the field arguments, we strictly follow the original definition of the Boulatov model~\cite{Boulatov:1992vp}. In later reformulations of the model this property is dropped while an additional combinatorial degree of freedom called color is attributed to the fields to guarantee that the perturbative expansion of the model is free of topological pathologies~\cite{Freidel:2009hd,Gurau:2010nd,Gurau:2009tw,Bonzom:2012hw}.}
 
\begin{figure}
 \subfloat[][Peter-Weyl mode \label{a}]{
  \def\svgwidth{0.2\columnwidth}
  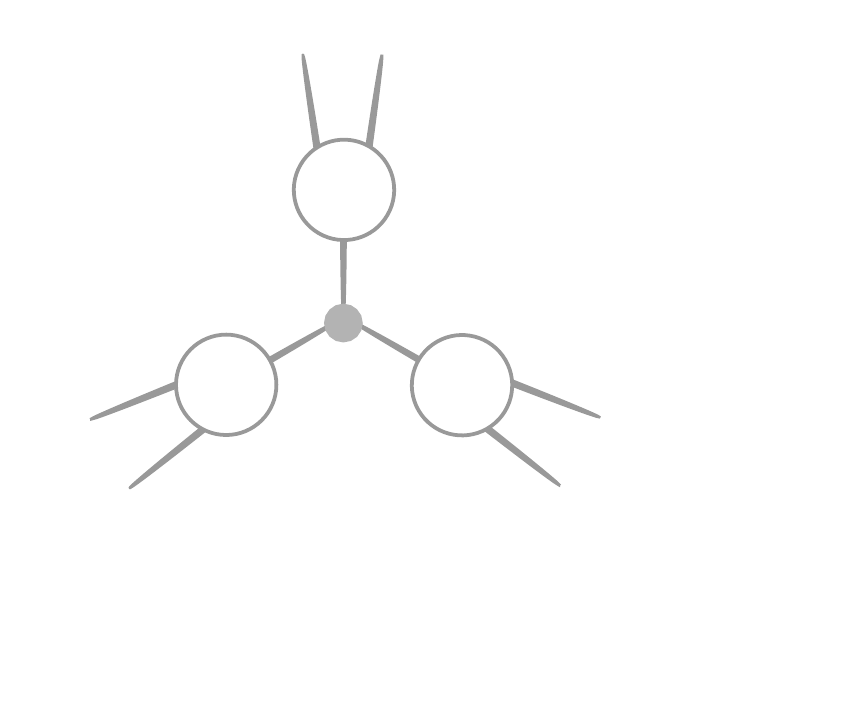
} \qquad
 \subfloat[][Mode with left invariance \label{b}]{
  \def\svgwidth{0.2\columnwidth}
  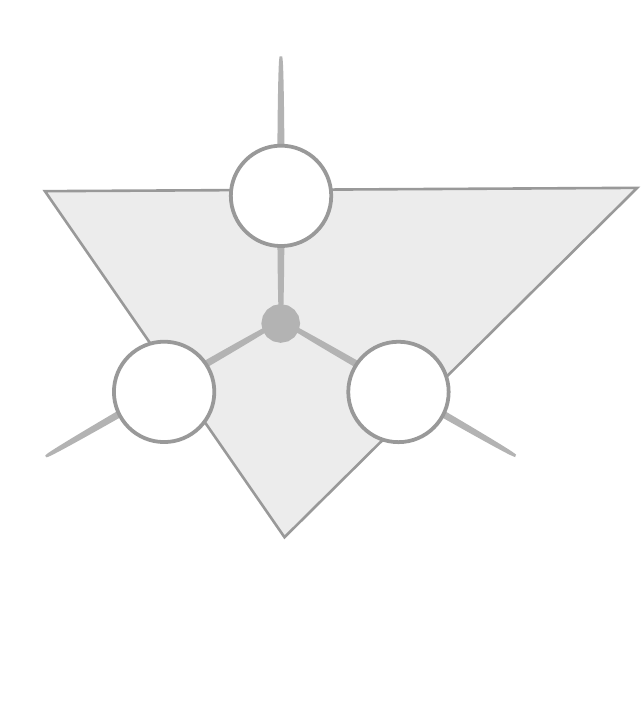
} \qquad
 \subfloat[][Mode with left and right invariance \label{c}]{
  \def\svgwidth{0.19\columnwidth}
  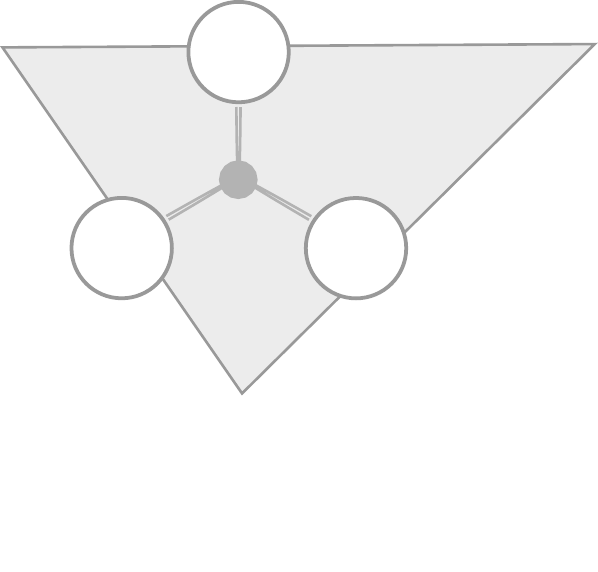
} \qquad
 \subfloat[][Equilateral mode with left and right invariance \label{d}]{
  \def\svgwidth{0.15\columnwidth}
  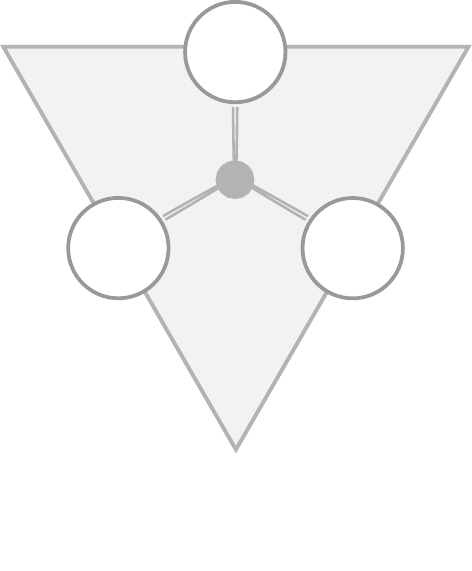
}
 \caption{Graphical representation of Peter-Weyl modes}
\label{pic:left-invariance}
\end{figure}

\medskip

\noindent{\bf The space $\mathcal{S}_{\textrm{EL}}$:}
To enforce rotational symmetry of the triangles, one requires left invariance of the field~\citep{Gielen:2014ila,Oriti:2016qtz}, such that for any $h\in\mathrm{SU}\left(2\right)$ the field $f$ satisfies
$
f\left(hx_{1},hx_{2},hx_{3}\right)=f\left(x\right)
$, (Fig.~\ref{d}).
Applications of GFT to quantum cosmology demonstrate that this symmetry is needed to identify the domain space of the fields with the superspace of homogeneous spatial geometries~\citep{Gielen:2014ila}.

The equilateral condition  on functions ensures  that their modes correspond to equilateral triangles. This condition can relate to isotropy in quantum cosmology studies of GFT, reflecting that we need to set all edges of the triangle to equal length in order to ensure equality in all directions (Fig.~\ref{d}).
It is crucial for the recovery of a Friedmann-like dynamics from GFT~\citep{Oriti:2016qtz,Oriti:2016ueo,Pithis:2016cxg,deCesare:2017ynn}.\footnote{Notice that for the subsequent analysis of extrema on $\mathcal{S}_{\textrm{EL}}$ the cyclicity property mentioned above has no impact and could in principle be lifted from the outset.}\footnote{The restriction to equilateral configurations bears strong resemblance to what is done in the closely related contexts of dynamical triangulations~\cite{Ambjorn:2001cv,Ambjorn:2013tki} and tensor models for quantum gravity~\cite{Gurau:2016cjo,gurau2016random} where the use of standardized building blocks --by universality arguments-- is believed not to affect the continuum results.}


Besides these arguments, there is also an algebraic reason to consider the restricted space $\mathcal{S}_{\text{EL}}$.
In GFT the action $S_{m,\lambda}$ defines statistical weights of
a generating functional using a functional integral~\eqref{Z}.
It has been shown, however, that on $\mathcal{S}$ the action $S_{m,\lambda}$
is not bounded from below, regardless of the parameter region~\citep{Freidel:2002tg,magnen2009scaling}. For this reason, the above integral is dominated by those field configurations that make the action $S_{m,\lambda}$
arbitrarily negative making Eq.~\eqref{Z} ill-defined. As we will show below, this problem gets resolved on $\mathcal{S}_{\text{EL}}$, where global minimizers of the action exist (at least for some parameter regions). This allows us to define  the generating functional perturbatively, and could lead to a well-defined statistical theory.\footnote{One can bound the Boulatov action by adding a so-called pillow term to the action~\cite{Freidel:2002tg}. We leave the impact of such a modification onto the ensuing analysis to future investigations.}

\section{Extrema and minimizers}
\label{sec:extremal condition}
\label{sec:extrema and minimizers}

Let $I\subset\mathbb{R}$ denote an
interval containing zero; for $t\in I$ and $\varphi,f\in\mathcal{\text{\ensuremath{\mathcal{S}_{\text{EL}}}}}$
a necessary condition for $\varphi$ to be a local minimizer on $\mathcal{S}_{\text{EL}}$
is given by
\begin{align}
S_{m,\lambda}^{'}\left(\varphi,f\right) &\doteq\partial_{t}S_{m,\lambda}\left(\varphi+tf\right)\vert_{0}=0,\label{eq:extremality} \\
S_{m,\lambda}^{''}\left(\varphi, f \right) &\doteq \partial_{t}^{2} S_{m,\lambda} \left(\varphi + t f \right)\vert_{0}\geq 0, \label{eq:neccessary condition for minimizers}
\end{align}
for any $f\in\mathcal{S}_{\text{EL}}$.

In the following we will  investigate the extremal condition~\eqref{eq:extremality} for the model~\eqref{eq:action}.
We will then check  if some solutions are minimal and thus fulfill~\eqref{eq:neccessary condition for minimizers} and the condition in definition~\ref{def:minimizer}.
\begin{prop}
\label{cor:equation of motion if Fourier coefficients}$\varphi\in\mathcal{S}_{\left(\rm{EL}\right)}$
is an extremum of $S$ if and only if the Peter-Weyl coefficients
of $\varphi$ --- denoted $\varphi^{J}$ --- satisfy for any $J\in\mathfrak{J}_{\left(\rm{EL}\right)}$,
\begin{align}
\label{FourierBoulatoveom}
 (J^{2}+m^{2})\varphi^{J}
 & +\frac{\lambda}{3!}\sum_{K \in \mathfrak{J}_{\left(\text{EL} \right)}
}\varphi^{j_{1}k_{2}k_{3}}\varphi^{j_{2}k_{3}k_{1}}\varphi^{j_{3}k_{1}k_{2}}\begin{Bmatrix}j_{1} & j_{2} & j_{3}\\
k_{1} & k_{2} & k_{3}
\end{Bmatrix}^{2} =0
\end{align}
where $K=(k_1,k_2,k_3) \in \mathfrak{J}_{\left(\text{EL} \right)}$.
\end{prop}
\begin{proof}
See appendix~\ref{subsec:proof for corollary 2}.
\end{proof}

The extremal condition~\eqref{FourierBoulatoveom} is a non-linear
tensor equation with an integral kernel given by the $6J$-symbol squared. To this issue adds the fact that the non-trivial zeros of the $6J$-symbol are still under investigation, making~\eqref{FourierBoulatoveom} inherently difficult to solve in full generality.
Some specific solutions for the case without the Laplace-Beltrami operator and $\lambda<0$ have been introduced in Ref.~\citep{Fairbairn:2007sv},
however, a systematic analysis of extrema was not performed therein.

Although the extremal condition~\eqref{FourierBoulatoveom}
is difficult to solve on $\mathcal{S}$,
it turns out to be  solvable on $\mathcal{S}_{\text{EL}}$,
because in this case the $6J$-symbol significantly simplifies.

\subsection{Extrema}
 \label{subsec:Extrema-of-the action}

In the following we will denote the Wigner $6J$-symbol for
$J\in\mathfrak{J}_{\left(\text{EL}\right)}$ by
$
\left\{ 6J\right\} \doteq\begin{Bmatrix}j_{1} & j_{2} & j_{3}\\
j_{1} & j_{2} & j_{3}
\end{Bmatrix},
$
and define the space $\mathfrak{J}_{\left(\text{EL}\right)}^{S}$
of  $J$'s such that
$
\mathfrak{J}_{\left(\text{EL}\right)}^{S}=\left\{ J\in\mathfrak{J}_{\left(\text{EL}\right)}\:\vert\:\text{with}\,\left\{ 6J\right\} \neq0\right\} .
$
In order to characterize the extrema of the action,
we define the space of extremal sequences. Let $C=\left(C^{J}\right)_{J\in\mathfrak{J}_{\left(\text{EL}\right)}}$
denote the sequence of (possibly complex) numbers such that for $J\in \mathfrak{J}^{S}_{\left( \text{EL} \right)}$
\begin{equation}
C^{J} \in \left\{ 0, \pm\frac{1}{\left|\left\{ 6J\right\} \right|}\sqrt{-\frac{3!}{\lambda}\left(J^{2}+m^{2}\right)} \right\}
\label{eq:coefficient solution}
\end{equation}
and for $J\in\mathfrak{J}_{\left(\text{EL}\right)}/\mathfrak{J}_{\left(\text{EL}\right)}^{S}$
\begin{equation}
C^{J} =
\begin{cases}
r \in \mathbb{R} & \text{if $J^2 = - m^2$} \\
0 & \text{otherwise}
\label{eq:coefficient solution outside of J}
\end{cases}
\end{equation}
Since $J^2 > 0$, the first case in~\eqref{eq:coefficient solution outside of J} can happen only when $m^2$ is negative and for $J \in \mathfrak{J}_{\text{EL}}$, $m^{2}$ has to be an even integer.
For simplicity,  we will exclude this case in the following analysis, because it requires a fine-tuning on the parameter $m^2$.
It is convenient to define the length $\ell$ of the sequence $C$
such that
\begin{equation}
\ell\left(C\right)=\sum_{J\in\mathfrak{J}_{\text{EL}}}\left|\text{sgn}\left(C^{J}\right)\right|,
\end{equation}
with the convention $\text{sgn}\left(0\right)=0$.
\begin{defn}\label{def:eml}
We define the space of extremal sequences as
\[
\mathcal{E}_{m,\lambda}=\left\{ C = \left( C^{J}\right)_{J\in \mathfrak{J}_{\text{EL}}}\,\vert\,C^{J}\in\mathbb{R},\,\ell\left(C\right)<\infty\right\} ,
\]
where the coefficients of each sequence are of the form~\eqref{eq:coefficient solution}.

\end{defn}
This space of course depends on the values of $m^{2}$ and $\lambda$, since  different choices of these parameters may violate the reality
condition $C^{J}\in\mathbb{R}$.
$\mathcal{E}_{m,\lambda}$ fully characterizes the space of extrema
of the action as states the following theorem.
\begin{thm}
\label{thm:extrema of dynamical Boulatov action}For any $C\in\mathcal{E}_{m,\lambda}$
the field $\varphi\in\mathcal{S}_{\text{EL}}$
\begin{equation}
\varphi\left(x\right)=\sum_{J\in\mathfrak{J}_{\textrm{EL}}}C^{J}\,\mathcal{X}^{J}\left(x\right)\label{eq:solutions}
\end{equation}
is an extremum of the action $S_{m,\lambda}$. Moreover, every equilateral
extremum of $S_{m,\lambda}$ is of the above form.
\end{thm}
\begin{proof}
See appendix~\ref{subsec:Proof-of-theorem 5}.
\end{proof}

We denote the space of extremal functions by $\mathcal{\tilde{E}}_{m,\lambda}$. It is worth mentioning that,  in spite of the non-linearity of the Euler-Lagrange equations, its solutions form a vector space over $\left(\mathbb{Z}_{3},+,\cdot\right)$.
\begin{cor*}
The space $\mathcal{\tilde{E}}_{m,\lambda}$ is a vector space over the discrete algebraic field $\left(\mathbb{Z}_{3},+,\cdot\right)$.
\end{cor*}

\begin{proof}
Denote the space of sequences with finitely many non-zero elements
over $\mathbb{Z}_{3}$ by $c_{00}\left(\mathbb{Z}_{3}\right)$.
It is a vector space over $\mathbb{Z}_{3}$.
Consider the map
\begin{align*}
\mathcal{I}: \mathcal{\tilde{E}}_{m,\lambda}
& \to c_{00}\left(\mathbb{Z}_{3}\right)\\
\varphi & \mapsto\left(\text{sgn}\left(C^{1}\right),\text{sgn}\left(C^{2}\right),\ldots\right),
\end{align*}
with the convention $\text{sgn}\left(0\right)=0$. $\mathcal{I}$  is one-to-one on its image, however, it may not be onto  $c_{00}\left(\mathbb{Z}_{3}\right)$
simply because the non-trivial zeros of the $6J$-symbol are not fully
characterized. Nevertheless, the image of $\mathcal{I}$ is algebraically closed and forms a subspace of $c_{00}\left(\mathbb{Z}_{3}\right)$.
For any $s=\left(s_{0},s_{1},\ldots\right)\in\mathcal{I}\left( \mathcal{\tilde{E}}_{m,\lambda}\right)$, the inverse mapping is given by
\begin{equation*}
\mathcal{I}^{-1}:s\mapsto
[\mathcal{I}^{-1}s](x)=
 \sum_{j\in\mathbb{N}}\text{sgn}\left(s_{j}\right)\,\left|C^{J_{j}}\right|\,\mathcal{X}^{J_{j}}\left(x\right),
\end{equation*}
where $J_j = (j,j,j)$, $j \in \mathbb{N}$,
with
\begin{equation}
\left|C^{J}\right| = \frac{1}{\left|\left\{ 6J\right\} \right|}\left|\sqrt{-\frac{3!}{\lambda}\left(J^{2}+m^{2}\right)}\right|.
\end{equation}
Since there are only finitely many non-zero coefficients, $s_{j}\neq0$,
the sum trivially converges in $\mathcal{S}_{\text{EL}}$. Since $\mathcal{I}$ is linear it is an isomorphism between $\mathcal{\tilde{E}}_{m,\lambda}$ and $\mathcal{I}\left( c_{00} \left( \mathbb{Z}_{3} \right)\right)$.

We define the sum on $\mathcal{\tilde{E}}_{m,\lambda}$ by
\begin{equation}
  \varphi_{1} +_{\mathbb{Z}_3} \varphi_{2} \doteq \mathcal{I}^{-1} \left( \mathcal{I}(\varphi_{1}) + \mathcal{I}(\varphi_{2})\right) .
\end{equation}
\end{proof}

 We now discuss the space of extremal sequences according to
 different  parameter regions, whose major
difference is captured by the sign of
the radicand in~\eqref{eq:coefficient solution}. We obtain the
four cases:
\begin{description}
\item [{$(a) \ m^{2}<0,~\lambda<0$}] the radicand is positive only if
\begin{equation}
J^{2}-\left|m^{2}\right| = 3j(j+1)-\left|m^{2}\right| \geq0,
\end{equation}
which is the case when $j$ satisfies
\begin{equation}
j_{\min} = \Bigg{\lceil}\,\frac{1}{6}\left(\sqrt{9+12|m^{2}|}-3\right)\,\Bigg{\rceil}\leq j,\label{eq:lower bound on spin}
\end{equation}
where $\lceil\cdot\rceil$ denotes the ceiling function. The space
of extremal sequences contains infinitely many sequences of the form
\[
\left(0,\ldots,0,C^{J_{\min}},C^{J_{\min}+1},\ldots\right),
\]
where we used the notation $J_{\min}+n\doteq\left(j_{\min}+n,j_{\min}+n,j_{\min}+n\right)$ for $n\in\mathbb{N}$,
with finitely many non-zero elements $C^{J}$.
\item [{$(b) \ m^{2}>0,~\lambda<0$}]
 all coefficients $C^{J}$ are real. The space of extremal sequences
can be written as
\begin{align*}
\mathcal{E}_{m,\lambda} & =\left\{ \left(C^{(0,0,0)},C^{(1,1,1)},\ldots\right)\,\vert\,\ell\left(C\right)<\infty\right\} .
\end{align*}
\item [{$(c) \ m^{2}>0,~\lambda>0$}]
the reality condition $C^{J}\in\mathbb{R}$ then requires $C^{J}=0$
for all $J\in\mathfrak{J}_{\text{EL}}$. The space of extremal sequences
contains a single zero-sequence
\[
\mathcal{E}_{m,\lambda}=\left\{ \left(0,0,0,\ldots\right)\right\} .
\]
\item [{$(d) \ m^{2}<0,~\lambda>0$}] the radicand is positive only if
\begin{equation}
3j\left(j+1\right)-\left|m^{2}\right|\leq0,\label{eq:upper bound}
\end{equation}
or equivalently for $j$ satisfying,
\begin{equation}
0\leq j\leq\Bigg{\lfloor}\,\frac{1}{6}\left(\sqrt{9+12\left|m^{2}\right|}-3\right)\,\Bigg{\rfloor}= j_{\max}\label{eq:Region for bounded J-1}
\end{equation}
where $\lfloor\cdot\rfloor$ denotes the floor function. In this case
$\mathcal{E}_{m,\lambda}$ contains finitely many sequences of the
form
\[
\left(C^{\left(0,0,0\right)},\ldots,C^{J_{\max}},0,0,\ldots\right),
\]
where $J_{\max}= \left( j_{\max},j_{\max},j_{\max} \right)\in\mathfrak{J}_{\text{EL}}$.
\end{description}

At this point, a few comments are in order: according to the geometrical interpretation in the previous section, each Fourier mode represents a triangle with the edge length $j$ and the area proportional to $J^{2}$. In the parameter regime $(d)$ relation~\eqref{eq:upper bound} provides an upper bound on the possible $j$'s for the extrema of the action. Hence, in this case, $\vert m^{2} \vert$ can be interpreted as the bound on the area of
the triangles determined by the extremal solutions.
This is an interesting geometrical fact that deserves further investigation.

A second remark is that the method of resolution restricting to equilateral configurations used to tackle ~\eqref{FourierBoulatoveom} certainly exports to GFT models
on higher dimensional manifolds $M=G^{\times D}$ with $G=\mathrm{SU}(2),\mathrm{SO}(4)$ and $D\in \mathbb{N}$. We expect that a similar result as in~\eqref{eq:coefficient solution} will hold if we replace the $6J$-symbol by the appropriate Wigner symbol and replace the square root by the $D-2$ root. However, the search of minimizers for these theories as performed in the subsequent analysis might be different.

\subsection{Minimizers}
 \label{subsec:Minimizer-of-the action}

We now seek the minimizers of the action and show that only two parameter regions admit global minimizers.

First, notice that in the case, $m^{2} <0, \ \lambda > 0$, the value of $\left| m^{2} \right|$ can determine, whether or not the action $S_{m,\lambda}$ is bounded from below.
To agree with this, assume the first non-trivial zero of the $6J$-symbol to be at $J_{0}\in \mathfrak{J}_{\text{EL}}$ and choose a function $f \left(x  \right) \doteq f^{J_{0}} \mathcal{X}^{J_{0}} \left( x \right)$ with $f^{J_{0}} \in \mathbb{R} $.
Then, for $\left| m^2 \right| > J_{0}^{2}$ the action evaluated at $f$ yields
\begin{equation}
S_{m,\lambda} \left(f\right)= \left( f^{J_{0}} \right)^{2} \ (J_{0}^{2} - \left| m^{2} \right|) < 0.
\end{equation}
Hence, the action can become arbitrarily negative and thus is unbounded from below. On the other hand, for $\left| m^{2} \right| < J_{0}^{2}$ the action has a global minimum as we will show in the following.

In order to give a general classification of solutions, we need to exclude cases when the $6J$-symbol vanishes.
A quick numerical analysis shows that for $\left| m^{2} \right| \leq 10^{9}$, the space of non-trivial zeros of the $6J$-symbol with $J^{2} \leq \left| m^{2} \right|$ is empty,
Therefore, theorem~\ref{thm:minimizeres} captures all possible solutions
up to this order.
In fact, we conjecture that for equilateral configurations, $\mathfrak{J}_{\text{EL}} / \mathfrak{J}_{\text{EL}}^{S} = \emptyset$, and our theorem holds for any value of $\left| m^{2} \right| $.

\begin{thm}
\label{thm:minimizeres}
Let $\left| m^2 \right|$ be such that for $j\leq j_{\max}$ every $J\in \mathfrak{J}_{\text{EL}}^{S}$ and such that there is no $J \in \mathfrak{J}_{\text{EL}}/\mathfrak{J}_{\text{EL}}^{S}$ such that $J^{2} - |m^{2}| = 0$. Then the equilateral extrema of the dynamical Boulatov action are of the following type:
\begin{enumerate}[label=(\alph*)]
\item For $m^{2}<0,\,\lambda<0$, all extrema are saddle points.
\item For $m^{2}>0,\,\lambda<0$, all non-trivial extrema are saddle points and the trivial extremum, $\varphi=0$, is a local minimizer on $\mathcal{S}_{\text{EL}}$.
\item For $m^{2}>0,\,\lambda>0$ the unique trivial extremum is a global minimizer
on $\mathcal{S}_{\text{EL}}$.
\item For $m^{2}<0,\,\lambda>0$ there are $2^{j_{\max}}$ global
minimizers on $\mathcal{S}_{\text{EL}}$ given by extremal sequences
$C\in\mathcal{E}_{m,\lambda}$ with maximal length, $\ell\left(C\right)=j_{\max}$.
Any other extremum of length $\ell\left(C\right)<j_{\max}$
is a saddle point.
\end{enumerate}
\end{thm}

\begin{proof}[Proof of theorem~\ref{thm:minimizeres}]
  In the following, let $\varphi \left(x\right)$ denote an extremum
  and let $f\in\mathcal{S}_{\text{EL}}$ be a generic function with
  the Peter-Weyl decomposition given by $f\left(x\right)=\sum_{J\in\mathfrak{J}_{\text{EL}}}f^{J}\,\mathcal{X}^{J}\left(x\right)$.
We remind here that a necessary condition for an extremum $\varphi \left( x \right)$ to be a minimizer (maximizer, respectively) is given by
\bea
S_{m,\lambda}^{''} \left( \varphi, f \right) & \geq 0
\qquad
\left(
S_{m,\lambda}^{''} \left( \varphi, f \right) \leq0 ,\; \text{resp.}
\right),
\eea
for any $f\in\mathcal{S}_{EL}$. In the Peter-Weyl decomposition the second variation
recasts as
\bea
S_{m,\lambda}^{''}\left(\varphi,f\right) & = \sum_{J\in\mathfrak{J}_{\text{EL}}}\left(f^{J}\right)^{2}
\left(
\vphantom{
\sum_{
K\in\mathfrak{J}^{S}_{ \text{EL} }}
} \left(J^{2} + m^{2} \right) -\frac{\lambda}{2}
\sum_{K \in\mathfrak{J}^{S}_{\text{EL}}} \delta_{J,K} \: \varphi^{K} \: \left\{ 6K \right\}^2 \right),
\eea
where $\varphi^{K}$ is the Peter-Weyl coefficient of the extremum $\varphi$.
The above condition is necessary but not sufficient. Nevertheless, it turns out to be useful to exclude some extrema.


\textit{Case (a)} ($m^2 \leq 0, \ \lambda \leq 0$):
By theorem~\ref{thm:extrema of dynamical Boulatov action}, extremal solutions contain only finitely many non-zero Fourier coefficients. Therefore it is possible to find $J_{>} \in \mathfrak{J}_{EL}$ such that $J_{>}^{2} - \left| m \right|  ^{2}>0$ and
$\varphi^{J_{>}}=0$.
Choosing
$
f_{>} \left( x \right)
\doteq
f^{J_{>}} \mathcal{X}^{J_{>}} \left( x \right)
$
the second variation gives
\begin{equation}
S_{m,\lambda}^{''}\left(\varphi, f_{>} \right) = \left( f^{J_{>}} \right)^{2} \left( J^{2}_{>} - \left| m^{2} \right| \right) > 0,
\end{equation}
which violates the maximizer condition.

To see that the minimizer condition is also violated, choose $f_{<} \left(x\right) \doteq f^{J_{<}} \mathcal{X}^{J_{<}} \left(x\right)$ such that $J_{<}^{2} - \left| m^{2} \right| < 0$. Then
the second variation is written as
\begin{equation}
S_{m,\lambda}^{''} \left( \varphi, f_{<} \right)
=
\left( f^{J_{<}} \right)^{2} \left( J_{<}^{2} - \left| m^{2} \right| \right) \leq 0.
\end{equation}
Hence, each  extremum in this parameter region violates the minimizer and the maximizer condition and therefore is a saddle point.

\textit{Case (b)} ($m^{2} \geq 0, \ \lambda \leq 0$):
For the non-trivial minimizer the above argument can also be applied in this case. Choosing the functions $ f_{>} \left( x \right)$ and $f_{<}  \left( x \right)$ as above we find
\be
S_{m,\lambda}^{''}\left(\varphi, f_{>} \right) = \left( f^{ J_{>} } \right)^{2} \left( J^{2}_{>}  + \left| m^{2} \right| \right) > 0,
\quad \text{ and } \quad
 S_{m,\lambda}^{''}\left(\varphi, f_{<} \right) = -2 \left( f^{J_{<}} \right)^{2} \left( J_{<}^{2} + \left| m^{2} \right| \right) < 0.
\ee
Hence, non-trivial extrema are saddle points. For the trivial extremum the second variation of $S_{m,\lambda}$ reads for any $f\in \mathcal{S}_{\text{EL}}$
\begin{equation*}
S^{''}_{m,\lambda} \left( 0, f \right) = \sum_{J\in \mathfrak{J}_{\text{EL}}} \left( f^{J} \right)^{2} \left( J^2 + \left| m^2 \right| \right) \geq 0,
\end{equation*}
and the necessary condition is satisfied. Indeed, the trivial extremum is a local minimum. To prove this we first notice that the Peter-Weyl transform is a topological isomorphism from $\mathcal{S}_{ \text{EL} }$ to the space of rapidly decreasing sequences $\mathcal{S} \left( \mathbb{N} \right)$ with topology given by the family of semi-norms~\citep[theorem 4]{sugiura1971},
\begin{align}
\| \left( f^{J} \right)_{J\in \mathfrak{J}_{\text{EL}}} \|_{n} = \sup_{J\in \mathfrak{J}_{\text{EL}}} \left| J^{n} f^{J} \right|.
\end{align}

The action evaluated at $f$ becomes
\bea
  S_{m,\lambda} \left(f \right) = \sum_{J\in \mathfrak{J}_{\text{EL}}} \left( f^{J} \right)^{2} \left( J^2 + m^2 \right)
    - \frac{\lambda}{4!} \sum_{J\in \mathfrak{J}_{\text{EL}}} \left( f^{J} \right)^{4} \left\{ 6J \right\}^2
\eea
Since the Wigner-$6J$-symbol is upper-bounded by $1$, we can estimate
\bea
   S_{m,\lambda}
   \left(f \right)   \geq
    \sum_{J\in \mathfrak{J}_{\text{EL}}} \left( f^{J} \right)^{2} \left( \left( J^2 + \left| m^2 \right| \right)
- \frac{\lambda}{4!} \left( f^{J} \right)^{2} \right)
\geq  \sum_{J\in \mathfrak{J}_{\text{EL}}} \left( f^{J} \right)^{2} \left(  m^2
  - \frac{\lambda}{4!} \left( f^{J} \right)^{2} \right). \label{eq:action at f}
\eea
Since Peter-Weyl transform is a topological isomorphism, we get for any $f\in \mathcal{S}_{\text{EL}}$ with $\| f \|_{0} \leq \sqrt{ \frac{4! m^2}{\left| \lambda \right|}}$, an estimate on the Fourier coefficients
\begin{equation}
\left| f^{J} \right| \leq \| \left( f^{J} \right)_{J\in \mathfrak{J}_{\text{EL}}} \|_{0} \leq \sqrt{\frac{4!m^2}{\left| \lambda \right|}}.
\end{equation}
Inserting this bound in~\eqref{eq:action at f} we obtain
$
S_{m,\lambda} \left( f \right) \geq 0 = S_{m,\lambda} \left( 0 \right).
$
Hence, in the neighborhood $N_{\epsilon,0} \cap \mathcal{S}_{\text{EL}}$ with
$\epsilon = \sqrt{\frac{4! m^2}{\left| \lambda \right|}}$ the trivial extremum is a minimizer.

\textit{Case (c)} ($m^{2}>0, \ \lambda > 0$):
In this case the space of extremal sequences contains only the zero-sequence, procuring  the trivial extremum $\varphi\left(x\right)=0$.
Denoting the quadratic part of the action in
\eqref{eq:action} by $Q_{m}\left(f\right)$ and the interaction
part by $\lambda I\left(f\right)$ such that
\begin{equation}
S_{m,\lambda}\left(f\right)=Q_{m}\left(f\right)+\lambda I\left(f\right),
\label{eq:action in Q and I}
\end{equation}
we have for any $f\in \mathcal{S}_{\text{EL}}$
\be
Q_{m}\left(f\right)  =\sum_{J\in\mathfrak{J}_{\text{EL}}}\left(f^{J}\right)^{2}\left(J^{2}+m^{2}\right)\geq0\,, \qquad
\lambda I\left(f\right) =\frac{\lambda}{4!}\sum_{J\in\mathfrak{J}_{\text{EL}}}\left(f^{J}\right)^{4}\,\left\{ 6J\right\} ^{2}\geq0.
\ee
Hence,
$S_{m,\lambda}\left(0\right)=0\leq S_{m,\lambda}\left(f\right)\,, \forall f\in\mathcal{S}_{\text{EL}}.
$
We obtain a global minimizer, since the minimal condition is satisfied on the whole $\mathcal{S}_{\text{EL}}$.

\textit{Case (d)} ($m^{2} < 0, \ \lambda >0$):
For any $f\in \mathcal{S}_{\text{EL}}$ the action evaluated at $f$ gives
\bea
S_{m,\lambda}\left(f\right)  =\frac{1}{2}\sum_{J\in\mathfrak{J}_{\text{EL}}}\left(f^{J}\right)^{2}\left(J^{2}-\left|m^{2}\right|\right)
+\frac{\lambda}{4!}\sum_{J\in\mathfrak{J}_{\text{EL}}}\left(f^{J}\right)^{4}\left\{ 6J\right\} ^{2}.
\eea
Splitting $f$ such that $f\left(x\right)=f^{-}\left(x\right)+f^{+}\left(x\right)$
with
\bea
f^{-}\left(x\right)  =\sum_{J \in\mathfrak{J}_{\text{EL}}}^{ \left|J\right|\leq3j_{\max}}f^{J}\,\mathcal{X}^{J}\left(x\right)\,,\qquad
f^{+}\left(x\right)  =\sum_{J\in\mathfrak{J}_{\text{EL}}}^{\left|J\right|>3j_{\max}}f^{J}\,\mathcal{X}^{J}\left(x\right),
\eea
we have
$
S_{m,\lambda}\left(f\right)= S_{m,\lambda} \left( f^{-} + f^{+} \right)\geq S_{m,\lambda}\left(f^{-}\right).
$
Hence, verifying the minimizer condition, it is enough
to show that
$
S_{m,\lambda}\left(\varphi\right)\leq S_{m,\lambda}\left(f^{-}\right).
$
The space of functions
of the form $f^{-}$ is finite-dimensional and we can use the usual
minimization procedure for functions. More specifically, let $s_{J}:\mathbb{R}\to\mathbb{R}$
be a function such that\begingroup \thinmuskip=1mu \medmuskip=2mu
minus 2mu \thickmuskip=3mu
\[
s_{J}\left(f^{J}\right)=\left(f^{J}\right)^{2}\left[\frac{1}{2}\left(J^{2}-\left|m^{2}\right|\right)+\frac{\lambda}{4!}\left(f^{J}\right)^{2}\left\{ 6J\right\} ^{2}\right].
\]
\endgroup
The action $S_{m,\lambda}\left(f^{-}\right)$ is smallest when each $s_{J}$ is minimal on $\mathbb{R}$ for each $J\leq J_{\max}$. Taking the first and second derivative of $s_{J}$ we see that the minimum is achieved by the coefficients $C^{J}$ from~\eqref{eq:coefficient solution}. Hence, an extremum given by an extremal sequence of maximal length is a global minimizer on the whole $\mathcal{S}_{\text{EL}}$.

If $\varphi$ is given by an extremal sequence $C$ of length $\ell\left(C\right)<j_{\max}$,
then there exists a $\mathcal{X}^{J_{0}}$ with $J_{0}\leq J_{\max}$
and $\varphi^{J_0}=0$. For $\delta \in \mathbb{R}$ define the function
$
g\left(x\right)=\varphi\left(x\right)+\delta\cdot\mathcal{X}^{J_{0}}\left(x\right).
$
Inserting $g$ into the action we get
\be
S_{m,\lambda}\left(g\right)  =S_{m,\lambda}\left(\varphi\right)
 +\delta^{2}\left[\frac{1}{2}\left(J_{0}^{2}-\left|m^{2}\right|\right)+\frac{\lambda}{4!}\delta^{2}\,\left\{ 6J_{0}\right\} ^{2}\right].
\ee
If $\delta^{2}$ is in the range $0<\delta<2C^{J_{0}}$ the square bracket
is negative and it follows
$
S_{m,\lambda} \left( g \right) \leq S_{m,\lambda} \left( \varphi \right).
$
Moreover, for any $\epsilon>0$ and $\delta<\frac{\epsilon}{J_{0}^{2n}}$ we
have
\be
\|g-\varphi\|_{n}  =\delta\,\sup_{x\in M}\left|\Delta^{n}\,\mathcal{X}^{J_{0}}\left(x\right)\right|=\delta\,\sup_{x\in M}\left|J_{0}^{2n}\,\mathcal{X}^{J_{0}}\left(x\right)\right|<\epsilon,
\ee
since the characters are bounded by one, $\left| \mathcal{X}^{J_0} \left( x \right) \right| \leq 1$. Hence, $g\in N_{\epsilon,n}\left(\varphi\right)$. For any $\epsilon >0$ choosing $\delta<\min\left(\frac{\epsilon}{J_{0}^{2n}},C^{J_{0}}\right)$
we get
$
S_{m,\lambda}\left(f\right)<S_{m,\lambda}\left(\varphi\right).
$
This shows that we can find a function $g$ in any neighborhood of $\varphi$ that decreases the value of the action, and hence, $\varphi$ is not a minimizer.
\end{proof}

\section{Concluding remarks}
\label{sec:discussion}

We investigated the minimizers of the dynamical Boulatov
action in four different parameter regions of the coupling constants. Our analysis is restricted to the space of smooth,  equilateral, left and right invariant functions, also invariant under cyclic permutations of its variables, $\mathcal{S}_{\text{EL}}$. This restriction ensures that the action is bounded from below for some parameter regions. Moreover, it is motivated by quantum cosmology studies on GFT.

 It appears that the very same restrictions allow us to solve the Euler-Lagrange equations for the dynamical Boulatov action and lead to a complete characterization of minimizers on the restricted space.  Our result characterizes the space of solutions by extremal sequences of finite length and shows that it forms a vector space over  $\mathbb{Z}_{3}$, which is surprising for the set of solutions to
a nonlinear integro-differential equation.
Furthermore, in the most interesting parameter region  $(d)$, the non-vanishing Fourier modes of extremal solutions are bounded by the coupling constant $m^{2}$, which suggests a connection between $m^{2}$ and the area of the triangle of the largest Peter-Weyl mode of the GFT field.

Our analysis shows that the region  $(a)$ does not have any minimizers on $\mathcal{S}_{\text{EL}}$, which makes this parameter region perhaps the least suitable for the definition of the statistical measure in~\eqref{Z}. For the parameter regions  $(b)$
and $(c)$ there is a single (local respectively global) minimizer given by the trivial extremum, $\varphi = 0$. Finally, in the region $(d)$ the action has $2^{j_{\max}}$ degenerate global minimizers, where $j_{\text{max}}$ is a function of the coupling constant $m^{2}$. The rich structure of global minima makes this region most interesting for further investigations, especially for the statistical theory.

On the space of equilateral functions only two possible parameter regions $(c)$ and
$(d)$  allow for the presence of global minimizers and hence could induce a meaningful definition of a non-perturbative statistical measure.

Case  $(c)$ admits a single global minimizer  $\varphi = 0$. Perturbation theory around this minimizer defines the perturbation theory in the coupling constant $\lambda$ and is used in the GFT literature to draw a connection to spin-foam models. Hence, our analysis would suggest that this regime is suitable for such a relation.

Case $(d)$, on the other hand, may suggest more structure for the quantum theory: a degenerate global minimum could lead to instantons or symmetry breaking in the corresponding statistical field theory in the following sense:

\textbf{Instantons:} The full non-perturbative formulation of a model is given by the minimizer of its quantum effective action. The latter is commonly assumed to be convex~\cite{Litim:2006nn} and therefore admits a single, unique minimizer. Hence, the difference between the minimizers of the classical and the quantum effective action becomes apparent, especially in the case when the classical action admits degenerate minimizers. In this case, a perturbative description around any of the minimizers of the classical action does not capture the non-perturbative effects of the theory. In quantum field theory, these non-perturbative effects can be understood as ``tunneling'' between the perturbative vacua, where the instanton action describes the tunneling probability. Thus, the degenerate structure of global minimizers in our case, suggests the necessity of instantons in the statistical formulation of GFT at least for the parameter region (d) (for a similar result see Ref.~\citep{Pithis:2016wzf}).

\textbf{Symmetry breaking:} this mechanism happens when the classical action admits degenerate global minimizers --- related by a symmetry of the classical action --- but the tunneling probability between them vanishes. As we already mentioned, the tunneling probability is described by the instanton action, which in ordinary field theory is often proportional to the volume of the base manifold. On a manifold with a finite volume, the tunneling probability is therefore finite. This often pertains to the statement that spontaneous symmetry breaking cannot occur in quantum field theories in a box. This realization, however, contains further assumptions that are satisfied in ordinary field theories but do not hold for GFT. It has been recently shown that even on the compact base manifold, $M=\mathrm{SU} \left( 2 \right)^{d}$ the tunneling between different perturbative minima can vanish~\cite{Kegeles:2017ems}, leading to a similar phenomenon of symmetry breaking.
In order to talk about symmetry breaking, we need to identify the symmetry, which in our case, is given by a flip of the sign of at least one of the modes in the Peter-Weyl decomposition of the minimizer (this can be modeled as a $\mathbb{Z}_2$-symmetry).
Since the action is of even power in the fields, such a flip will not affect the value of the action and will correspond to a discrete symmetry.
For this reason, it is possible that the global minimizers of the action provoke the breaking of sign-flip symmetry. This needs to be investigated more rigorously in future work.

For ordinary local quantum field theories, a symmetry breaking mechanism can sometimes be related to a phase transition and the formation of a condensate. In particular, this could be the signal of a Bose-Einstein condensation just as expected for quantum cosmology studies in GFT.
A closer look at the solutions found for sector (d) shows that these might bear intriguing perspectives. Indeed, the  `particle' number, used in the condensate cosmology context, is computable in terms of the $L^{2}$-norm of the minimizer. In the present situation, that very number
proves to be bounded by the parameter $m^{2}$:
\be
N \doteq \| \varphi \|_{L^{2}} = \frac{3!}{\lambda} \sum_{J\in \mathfrak{J}_{\text{EL}}^{S}}^{\left| J \right| \leq 3 j_{\text{max}}} \frac{1}{\left\{ 6J \right\}^{2}} \left| J^{2} - \left| m^{2} \right| \right|
 \leq \frac{3! \left| m^{2} \right| }{\lambda}
  \sum_{J\in \mathfrak{J}_{\text{EL}}^{S}}^{\left| J \right| \leq 3 j_{\text{max}}}
   \frac{1}{\left\{ 6J \right\}^{2}}
 \leq \frac{3! \left| m^{2} \right| }{\lambda} \ j_{\text{max}} \  C_{\text{max}},
\ee
with $C_{\text{max}}=\max_{J\in \mathfrak{J}_{\text{EL}}^{S}}
\left( \left\{ 6J \right\}^{-2} \right) $.
For $\left| m^{2} \right| \gg 1$ we can approximate $j_{\text{max}}$ further as $j_{\text{max}}\leq 2\left| m^{2} \right|$ and obtain a simpler bound on  the $L^{2}$-norm of the minimizers
\begin{equation}
  N \leq \frac{12 \left| m^{4} \right|}{\lambda} \ C_{\text{max}}.
\end{equation}
The coupling constant $m^{2}$ (or $|m^{4}|/\lambda\gg 1$) could be large but that itself is not enough to ensure $N=\infty$. Nevertheless, starting from our solutions, a divergent parameter $m^{2}$ is a necessary condition for the divergent $L^{2}$-norm. A large particle number would be desirable for the condensate cosmology approach because such configurations could then be interpreted as to define non-trivial homogeneous and isotropic background geometries in 3d with Euclidean signature. This point deserves further investigations.

We should mention here that our analysis does not capture minimizers with a divergent $L^{2}$-norm
(dealing only with integrable functions), and some modifications will be in order to consider these cases.
 One necessary modification would be to relax the smoothness condition of the minimizers and use the space of tempered distributions instead.
 This could be particularly interesting for GFT models without the Laplace-Beltrami operator, which correspond to a topological BF-theory.  Due to the distributional nature of minimizers their
   $L^{2}$-norm will sometimes diverge making them potentially interesting
 for quantum cosmological studies~\citep{Pithis:2016wzf,Kegeles:2017ems} and spin-foam models~\cite{Fairbairn:2007sv}. The solutions to these GFT models must be addressed differently
but certainly deserve further attention.

There are several models using tensor fields (with interesting properties
such as perturbative renormalizability) which do not impose strong
symmetry conditions on the fields.
These models' interactions could also be radically different from that of Boulatov~\cite{Geloun:2013saa}.
Their corresponding  Euler-Lagrange equation (without $6J$-symbols) still involves a non-linear tensor like equation, and it remains a difficult task to solve them.
In this case, an approach to circumvent the non-linearity and to obtain solution fields which are more general than equilateral configurations is to consider symmetric tensor fields and to decompose the field into its traceless part and the rest, namely vector-like components~\cite{hamermesh2012group}.  Such a decomposition could help to solve the extremal conditions on $\mathcal{S}$ which might find applications in GFT studies of inhomogeneous and anisotropic quantum cosmologies.

On the other hand, the existence of global minima on $\mathcal{S}_{\text{EL}}$ suggests that we can define a self-consistent statistical theory using only this space. This theory could potentially be well-defined due to the bound of the action on $\mathcal{S}_{\text{EL}}$ and may have implications for quantum cosmological studies of GFT.

\subsection*{Acknowledgments}
The authors thank L. Freidel, S. Gielen, E. Livine, D. Oriti and J. Th\"urigen for useful remarks. AP and AK gratefully acknowledge the constant hospitality at the Max-Planck-Institute for Gravitational Physics, Albert Einstein Institute. JBG thanks the Laboratoire de Physique Th\'eorique d'Orsay for its warm hospitality.

\appendix

\section*{Appendix}

\section{Harmonic analysis on $\mathrm{SU}(2)$}
\label{sec:appendix}

This appendix gathers the main identities on the harmonic analysis
on $\mathrm{SU}(2)$ repeatedly used throughout the text.

\subsection{Peter-Weyl transform \label{sec:Peter-Weyl decomposition}}

We briefly recall
the most important properties of the Peter-Weyl transform and
 Wigner matrices, needed for the  harmonic
analysis on $\mathrm{SU}(2)$. Let  $\mathcal{C}^{\infty}\left(\textrm{SU}(2)\right)$  be
the space of smooth functions $f$ on $\textrm{SU}(2)$ which is equipped
with the topology given by semi-norms
\begin{equation}
\|f\|_{n}=\sup_{g\in\textrm{SU}(2)}\left|\Delta^{n}f\left(g\right)\right|,
\end{equation}
with $\Delta$ the Laplace-Beltrami operator and $n\in\mathbb{N}$.

For any $f\in  \mathcal{C}^{\infty}\left(\mathrm{SU}\left(2\right)\right)$
there exists a sequence of complex numbers $\left( f_{mn}^{j}\right) $
with $j\in\frac{\mathbb{N}}{2}$ and $m,n\in\left\{ -j,\ldots,j\right\} $
and $D_{mn}^{j}\left(x\right)$ denote the Wigner matrix coefficients
with $d_{j}=2j+1$ such that
\begin{equation}
\lim_{N\to\infty}\sum_{j=0}^{N}\sum_{m,n=-j}^{j}\sqrt{d_{j}}f_{mn}^{j}D_{mn}^{j}=f,
\end{equation}
in the above topology. The sequence of Fourier coefficients $\left( f_{mn}^{j}\right) $
is rapidly decreasing, i.e.  for any $K\in\mathbb{N}$
\begin{equation}
\sup_{j}\left|j^{K}\sum_{\alpha,\beta=-j}^{j}\bar{f}_{\alpha\beta}^{j}f_{\alpha\beta}^{j}\right|<\infty.\label{eq:Seminorms}
\end{equation}
If we call the space of rapidly decreasing sequences $\mathcal{S}\left(\mathbb{N}\right),$
then~\eqref{eq:Seminorms} defines a family of semi-norms
on $\mathcal{S}\left(\mathbb{N}\right)$ and in the corresponding
topology it becomes a Fr\'echet space. Then the
 Peter-Weyl transform $\mathcal{F} :\mathcal{C}^{\infty}\left(\mathrm{SU}\left(2\right)\right)\to\mathcal{S}\left(\mathbb{N}\right)$
is a topological isomorphism between the space of smooth functions
and the space of rapidly decreasing sequences~\citep{sugiura1971}.

In our work, we deal with functions on three copies of $\textrm{SU}(2)$.
For this reason, we introduce $M=\textrm{SU}(2)^{\times3}$
as a Lie group with points $\left(x_{1},x_{2},x_{3}\right)$. The
representations of $M$ are given by product representations such
that
\begin{align}
\mathfrak{D}^{\left(j_{1},j_{2},j_{3}\right)}:\textrm{SU}(2)^{\times3} & \to L\left(V^{j_{1}}\otimes V^{j_{2}}\otimes V^{j_{3}}\right)
\end{align}
with
$
\mathfrak{D}^{(j_{1},j_{2},j_{3})}=D^{j_{1}}\otimes D^{j_{2}}\otimes D^{j_{3}},
$
where $L \left( V \right)$ denotes the space of linear maps on $V$
 a  vector space.

It follows by the Peter-Weyl theorem that the matrix coefficients
$\mathfrak{D}_{\alpha,\beta}^{J}\left(x\right)$ are dense in the
space of smooth functions on $M$, where now $J,\alpha$ and $\beta$
are multi-indices such that $J=\left(j_{1},j_{2},j_{3}\right)$ with
$j_{1},j_{2},j_{3}\in\frac{\mathbb{N}}{2}$ and $\alpha=\left(\alpha_{1},\alpha_{2},\alpha_{3}\right),\beta=\left(\beta_{1},\beta_{2},\beta_{3}\right)$
such that $\alpha_{i},\beta_{i}\in\left\{ -j_{i},\ldots,j_{i}\right\} $
for $i\in\left\{ 1,2,3\right\} $.

\subsection{Basis for left and right invariant functions\label{subsec:Basis for left and right invariant functions}}

In the above notations,
the left and right invariant functions  on $M=\textrm{SU}(2)^{\times3}$ are given by group
averaging, such that for any $f\in\mathcal{C}^{\infty}\left(M\right)$,
 and any $x =(x_1,x_2,x_3)\in M$,
$
\int\text{d}L\text{d}R\:f\left(Lx_{1}R,Lx_{2}R,Lx_{3}R\right),
$
with $L, R \in \text{SU(2)}$.
 In the Peter-Weyl decomposition
a left and right invariant function $f$ assumes the form
\be\label{eq:PW-decomposition for left and right invariance}
f\left(x\right)  =\sum_{J}\,f^{J}\,\sqrt{d_{j_{1}}d_{j_{2}}d_{j_{3}}}\,
 \int\text{d}h\:\chi^{j_{1}}\left(x_{1}h\right)\chi^{j_{2}}\left(x_{2}h\right)\chi^{j_{3}}\left(x_{3}h\right),
\ee
where  $J=\left(j_{1},j_{2},j_{3}\right) \in (\frac{\mathbb{N}}{2})^{\times 3}$,
and $\chi^{j_{i}}$ denotes the character of the representation of $\mathrm{SU} \left( 2 \right)$ with
dimension $d_{j_{i}}$. We denote the integral of the product of three characters
by
\be
\mathcal{X}^{J}\left(x\right)  \doteq\sqrt{d_{j_{1}}d_{j_{2}}d_{j_{3}}}\,
 \int\text{d}h\:\chi^{j_{1}}\left(x_{1}h\right)\chi^{j_{2}}\left(x_{2}h\right)\chi^{j_{3}}\left(x_{3}h\right).
\ee
It can be easily checked that  $\mathcal{X}^{J}$
has the following properties:
\begin{enumerate}
\item Using the orthogonality of characters, $\int \text{d}h~\chi^{j}(hx_{1})\chi^{l}(x_{2}h)=\frac{\delta_{jl}}{d_{j}}\chi^{j}(x_{2}x_{1}^{-1})$, and reality of  characters,
the $\mathcal{X}^{J}$'s are real-valued and form an orthonormal family with respect
to the $L^{2}\left(M,\text{d}x\right)$ scalar product:
\begin{equation}
\int_{M}\text{d}x\;\mathcal{X}^{J}\left(x\right)\mathcal{X}^{K}\left(x\right)=\delta_{J,K};
\end{equation}
\item  $\mathcal{X}^{J}$ is proportional to the $3J$-Wigner symbol with three equal $j$'s and sum over the magnetic indices and hence vanishes  if $j$ is not an integer;
\item Using~\eqref{eq:PW-decomposition for left and right invariance},
 the family of $\mathcal{X}^{J}$'s is  dense in the space of left and right invariant functions,
such that any left and right invariant function $f$ can be written as
\begin{equation}
f\left(x\right)=\sum_{J \in \mathfrak{J}}\,f^{J}\,\mathcal{X}^{J}\left(x\right);
\end{equation}

\item Using the fact  that the Wigner $6J$-symbol can be defined in terms of characters as~\citep{biedenharn_louck_carruthers_1984},
\begin{equation}
\begin{Bmatrix}l_{01} & l_{02} & l_{03}\\
l_{23} & l_{13} & l_{12}
\end{Bmatrix}^{2}=\int(\text{d}h)^{4}\prod_{i<j}^{3}\chi^{l_{ij}}(h_{j}h_{i}^{-1}).\label{eq:6j and the characters}
\end{equation}
the $6J$-symbol is given by a $\mathcal{X}^{J}$ integral as
\bea
\delta_{j_{1},k_{1}} \delta_{j_{2},l_{1}} \delta_{j_{3},q_{1}} \delta_{q_{2},l_{3}} \delta_{k_{2},q_{3}} \delta_{k_{2},l_{3}}
\begin{Bmatrix}
j_{1} & j_{2} & j_{3}\\
q_{2} & l_{2} & k_{2}
\end{Bmatrix}^{2}
 =\int\text{d}x\text{d}y\text{d}z\text{d}w\,\text{Tet}\left(x,y,z,w\right) \label{eq:6j with mathcal X}
 \mathcal{X}^{J}\left(x\right)\mathcal{X}^{K}\left(y\right)\mathcal{X}^{L}\left(z\right)\mathcal{X}^{Q}\left(w\right) \nonumber,
\eea
with
 $J= (j_1,j_2,j_3), \ K=(k_1,k_2,k_3),\ L=(l_1,l_2,l_3)$, $Q=(q_1,q_2,q_3)$.

\end{enumerate}
Since we are interested in functions that are invariant under cyclic permutation we need to symmetrize the characters $\mathcal{X}^{J} \left( x \right)$.
To achieve this, we introduce the symmetrization operator
\begin{equation}
 P\mathcal{X}^{J} \left( x \right)= \frac{1}{3} \sum_{\sigma \in \text{Cyc}} \mathcal{X}^{\left( j_{\sigma{(1)}}, j_{\sigma{(2)}}, j_{\sigma{(3)}} \right)} \left( x \right),
\end{equation}
where Cyc denotes the set of cyclic permutations of $\{1,2,3\}$.  All aforementioned properties of $\mathcal{X}^{J}$ can be adapted to $P\mathcal{X}^{J} \left( x \right)$ by including a normalized sum over cyclic permutations of indices. Since for the equilateral case we have $P\mathcal{X}^{J} \left( x \right) = \mathcal{X}^{J} \left( x \right)$, we simply use the notation $\mathcal{X}^{J} \left( x \right)$ for symmetric characters on $\mathcal{S}_{\text{EL}}$ and on $\mathcal{S}$.

\section{Proofs}

\subsection{Proof of proposition~\ref{cor:equation of motion if Fourier coefficients}}
\label{subsec:proof for corollary 2}
Consider the action $S_{m,\lambda}$~\eqref{eq:action} and $S'_{m,\lambda}$~\eqref{eq:extremality};
$\mathcal{S}_{\left(\text{EL}\right)}$
means either $\mathcal{S}$ (space of right invariant functions) or $\mathcal{S}_{\text{EL}}$
(space of left and right invariant and equilateral functions).
The following statement holds:
\begin{lem}
\label{lem:variation in D's}The field $\varphi\in\mathcal{S}_{\left(EL\right)}$
is an extremum of $S_{m,\lambda}$ iff
\begin{equation}
S_{m,\lambda}^{'}\left(\varphi,\mathcal{X}^{J}\right)=0
\end{equation}
for all $J\in\mathfrak{J}_{\left(EL\right)}$.
\end{lem}
\begin{proof}
Let $\varphi$ be an extremum of $S_{m,\lambda}$, then the ``only if'' direction
is obvious since for any $J\in\mathfrak{J}_{\left(\text{EL}\right)}$
the functions $\mathcal{X}^{J}$ are in $\mathcal{S}_{\left(\text{EL}\right)}$.

For the ``if'' direction we observe the following: since the set
$\left\{ \mathcal{X}^{J}\right\} _{J\in\mathfrak{J}_{\left(\text{EL}\right)}}$
is dense in $\mathcal{S}_{\left(\text{EL}\right)}$, for any $f\in\mathcal{S}_{\left(\text{EL}\right)}$
there exists a family of real numbers $\left\{ f^{J}\right\} _{J\in\mathfrak{J}_{\left(\text{EL}\right)}}$
such that the sequence of functions given for all $N \in \mathbb{N}$ as
\begin{equation}
f_{N}\left(x\right)=\sum_{J\in\mathfrak{J}_{\left(\text{EL}\right)}}^{\left|J\right|<N}f^{J}\,\mathcal{X}^{J}\left(x\right),
\end{equation}
converges to $f$. Then
$
c=\sup_{x\in M}\sup_{N\in\mathbb{N}}\left|f_{N}\left(x\right)\right|,
$
exists and dominates each $f_{N}$ such that, $\left|f_{N}\right|\leq c$.
Moreover, $c$, seen as a constant function on $M$, is integrable
since $M$ is compact.

For any $f\in\mathcal{S}_{\left(\text{EL}\right)}$ the extremal condition
for the action $S_{m,\lambda}$ reads as
\be
S^{'}_{m,\lambda}\left(\varphi,f\right)  =\int_{M}\text{d}x\:f\left(x\right)\left(-\Delta+m^{2}\right)\varphi\left(x\right)
  +\frac{\lambda}{3!}\int_{M^{\times4}}\text{d}x\text{d}y\text{d}z\text{d}w~\text{Tet}\left(x,y,z,w\right)
  f\left(x\right)\varphi\left(y\right)\varphi\left(z\right)\varphi\left(w\right).
\ee
Using the Peter-Weyl decomposition for $f$, we can interchange
the limit and the integral by the dominant convergence theorem
(using the bound $c$) and
obtain
\begin{align*}
S'_{m,\lambda}\left(\varphi,f\right) & =\lim_{N\to\infty}\sum_{J}\,f^{J}\,S^{'}_{m,\lambda}\left(\varphi,\mathcal{X}^{J}\right)=0,
\end{align*}
for any $f\in\mathcal{S}_{\left(EL\right)}$, from which the statement
follows.
\end{proof}

\begin{cor*}
$\varphi\in\mathcal{S}_{\left(\text{EL}\right)}$ is an extremum of
$S$ if and only if the Peter-Weyl coefficients of $\varphi$ ---
denoted by $\varphi^{J}$ --- satisfy for any $J\in\mathfrak{J}_{\left(\text{EL}\right)}$,
\be
 (J^{2}+m^{2})\varphi^{J}
 +\frac{\lambda}{3!}\sum_{\substack{k_{i}}
}\varphi^{j_{1}k_{2}k_{3}}\varphi^{j_{2}k_{3}k_{1}}\varphi^{j_{3}k_{1}k_{2}}\begin{Bmatrix}j_{1} & j_{2} & j_{3}\\
k_{1} & k_{2} & k_{3}
\end{Bmatrix}^{2}=0.
\ee
\end{cor*}
\begin{proof}
From lemma~\ref{lem:variation in D's} the extremal condition is given
by the variation in the basis direction $\mathcal{X}^{J}$ for any $J\in\mathfrak{J}_{\left(\text{EL}\right)}$.
Inserting the Peter-Weyl decomposition of $\varphi$ in the action
$S_{m,\lambda}\left(\varphi\right)$, interchanging the limit with
the integral by the dominant convergence theorem and using the relation
in~\eqref{eq:6j with mathcal X} we obtain the desired statement.
 \end{proof}

\subsection{Proof of theorem~\ref{thm:extrema of dynamical Boulatov action}\label{subsec:Proof-of-theorem 5}}
\begin{thm*}
For any $C\in\mathcal{E}_{m,\lambda}$ the field $\varphi\in\mathcal{S}_{\text{EL}}$
\begin{equation}
\varphi\left(x\right)=\sum_{J\in\mathfrak{J}_{\text{EL}}}C^{J}\,\mathcal{X}^{J}\left(x\right)\label{eq:solutions-1}
\end{equation}
is an extremum of the action $S_{m,\lambda}$. Moreover, every equilateral
extremum of $S_{m,\lambda}$ in $\mathcal{S}_{\text{EL}}$ is of the
above form.
\end{thm*}
\begin{proof}
To show that $\varphi$ solves the extremal condition we need to show,
by proposition~\ref{cor:equation of motion if Fourier coefficients},
that each $C^{J}$ satisfies~\eqref{FourierBoulatoveom},
which follows by direct calculation.

Conversely, every equilateral function can be written as
\begin{equation}
f\left(x\right)=\sum_{J\in\mathfrak{J}_{\text{EL}}}A_{J}\,\mathcal{X}^{J}\left(x\right),
\end{equation}
with $\left(A^{J}\right)_{J\in\mathfrak{J}_{\text{EL}}}$ being a
rapidly decreasing sequence~\cite{sugiura1971}. Using proposition
\ref{cor:equation of motion if Fourier coefficients}, we find that
the extremal solutions have coefficients $A^{J}$ which satisfy
\begin{equation}
A^{J}\in \left\{\pm\frac{1}{\left|\left\{ 6J\right\} \right|}\sqrt{-\frac{3!}{\lambda}\left(J^{2}+m^{2}\right)},0 \right\},
\end{equation}
or $A^{J}\in \mathbb{R}$ for $J\in \mathfrak{J}_{\text{EL}} / \mathfrak{J}_{\text{EL}}^{S}$ with $J^2 + m^2 =0$.
If $A^{J}$ is not trivial we can estimate its growth using the asymptotic behavior of $6J$-symbols~\cite{Gurau:2008yh} as
\begin{equation}
A^{J}\sim\frac{j}{\left|\left\{ 6J\right\} \right|}\sim j^{\frac{5}{2}}.
\end{equation}

However, for $\left(A^{J}\right)_{J\in\mathfrak{J}_{\text{EL}}}$
to be a rapidly decreasing sequence, the coefficients have to satisfy
for any $n\in\mathbb{N}$,
\begin{equation}
\lim_{j\to\infty}\left|j\right|^{n}\left|A^{J}\right|\to0.
\end{equation}
This is only possible if $A^{J}=0$ for all but finitely
many $J\in\mathfrak{J}_{\text{EL}}$.
\end{proof}


\end{document}